\documentclass[nips13submit_09,times,art10]{article} 
\usepackage{nips13submit_e,times}
\usepackage{hyperref}
\usepackage{url}
\usepackage{amsmath}
\usepackage{amsthm}
\usepackage{amssymb}
\usepackage{graphicx}

\newtheorem{thm}{Theorem}

\title{An Entropy Maximizing Geohash for Distributed Spatiotemporal Database Indexing}

\author{
Taylor~B.~Arnold \\
AT\&T Labs Research\\
33 Thomas Street\\
New York, NY 10007 \\
\texttt{taylor@research.att.com}
}

\nipsfinalcopy 

\begin{document}

\maketitle

\begin{abstract}
We present a modification of the standard geohash algorithm
based on maximum entropy encoding in which the data volume
is approximately constant for a given hash prefix length.
Distributed spatiotemporal databases, which typically require interleaving
spatial and temporal elements into a single key, reap large benefits
from a balanced geohash by creating a consistent ratio between spatial
and temporal precision even across areas of varying data density.
This property is also useful for indexing purely spatial datasets,
where the load distribution of large range scans is an imporant aspect
of query performance.  We apply our algorithm to data generated
proportional to population as given by census block population
counts provided from the US Census Bureau.
\end{abstract}

\section{Introduction}  \label{sec:intro}

There has been a rapid increase in petabyte-scale data sources across
a range of industry, government, and academic applications. Many of
these sources include data tagged with both geospatial and temporal
attributes. Website HTTP request logs, along with geocoded IP addresses,
are a common example where even modestly sized sites can quickly produce
large-scale datasets. Inexensive embedded GPS devices are another common
source, and are typically found in everything from automobiles to cellphones.
Sensor data from Radio-frequency identification (RFID) devices, which have
been readily adapted to monitor complex supply-chain management operations,
amongst other applications, also generate large datasets with both location
and time-based features.

General purpose distributed filesystems such as the Google Filesystem \cite{ghemawat2003google}
and the Hadoop Filesystem \cite{shvachko2010hadoop}, provide adequate
support for the raw storage of geospatial data sources. These alone however only
provide minimal support for querying file chunks. Analysis, aggregations, and filtering
often require processing through a large subset of the raw data. Database functionality
is provided by additional software, often modelled off of Google's
BigTable \cite{chang2008bigtable} design. Data are stored lexicographically
by a primary key; queries take the form of a scan over a range of contiguous
primary keys along with optional, additional query filters.
Popular open source implementations of this form include Amazon's DynamoDB \cite{decandia2007dynamo},
Apache Accumulo \cite{fuchs2012accumulo}, Apache HBase \cite{taylor2010overview},
and Apache Cassandra \cite{lakshman2010cassandra}.

For geospatial data sources which are typically queried on an entity--by--entity basis, the
standard distributed database applications are straightforward to implement. A uniquely
identifying serial number for each entity can be used as the primary key, with spatial
and temporal filtering accomplished with a combination of in-database filters; a single
entity will typically constitute only a negligable proportion of the overall data volume.
This application model for example would satisfy the needs of a user-facing supply chain
system in which end-users search for the current location of an expected delivery.
Unfortunately, more complex spatial queries do not fit neatly into the distributed database
design. A single primary key is natively one-dimensional, whereas geospatial data
is typically two-dimensional and, with the addition of a temporal component, generally
no less than three. A common method for alleviating the dimensionality problem of having a single primary
index is through the use of a space-filling curves such as Peano, Gosper, and Dragon
curves. A geohash is a particular algorithm which uses a space-filling curve for
encoding latitude and longitude tuples as a single string. It has had fairly wide adoption
for this purpose; for example, it has been included as a core functionality in the popular
distributed document store MongoDB \cite{chodorow2013mongodb}.

The geohash concept adapts well for purely geospatial distributed datasets, allowing for
querying spatial regions of vastly different scales through the use of a single index; larger
regions specify a longer range of geohashes whereas smaller regions specify a shorter range.
A temporal component can be incorporated into the space-filling curve underlying a standard
geohash, however this requires an unavoidable, fixed database design decision as to the
appropriate scaling factor between space and time. For instance: should a square kilometre
bucket in space coorispond to a month, day, or hour bucket in time? The decision has
important performance implications as spatial queries take the form of a collection of
range queries of primary keys. If the buckets are not appropriately sized queries will either
include many false positives (data returned outside of the area of interest) or require
stitching together a large number of small scans. Either of these can cause significant
performance deterioration, due to additional I/O time in the case of false positives and
additional overhead in running a large number of small jobs.

In many cases the question of how to balance space and time into a single key is complicated
by the fact that data are distributed very unevenly over the spatial dimensions. Sources
which are generated roughly proportional to population density will produce 4 orders of
magnitude or more in urban environments compared to rural ones. In these common cases, it
is often more natural to scale the spatial buckets to be proprotional to the data volume.
In this case, we instead only need to decide: should a bucket with 1 million data points
coorispond to month, day, or hour bucket in time? This question will typically have a more
natural answer for a given database system than the previous one, particularly when the
query results are used as inputs to complex algorithms such as recommender systems or
anomaly detection routines.

In order to implement a data balanced alternative to the standard space filling curves,
we present a novel entropy maximizing geohash. Our method is necessarily model based and
learned from a small subset of spatial data. It is robust both to random noise and
model drift. It can implemented with only a modest increase in computational complexity
and is generic enough to be implemented as-is in conjuction with several recently
proposed geohash based schemes for storing large geospatial--temporal datasets. As a
side effect of the entropy based encoding it has also reduces the size of the raw database
files; this has been observed even after applying aggressive compression techniques.

The remainder of this article is organized as follows: Section~\ref{sec:form}
provides a specific formulation of the problems we are trying to solve and the
deficiencies of addressing them using a the standard geohash.
Section~\ref{sec:related} provides a brief literature review of geospatial
database techniques with a particular focus on how these can be used in conjuction
with our methods.
In Section~\ref{sec:geohash} we give a thorough mathematical description of the standard
geohash algorithm and present the formulation for entropy maximizing geohash.
In particular, a theoretical result establishes an upper bound on how well the
maximized geohash can be learnt from a sample of spatial data. Section~\ref{sec:census} applies the
entropic geohash to data from the US Census Bureau, empirically
demonstrating how robust the method is to model drift and Section~\ref{sec:future}
provides avenues for future extensions.

\section{Problem Formulation} \label{sec:form}

Consider partitioning an area of space into equally sized square boxes. If these
regions were given a unique identifier, a basic geospatial-temporal index can be constructed
by concatinating the box id with a timestamp. Querying a region in spacetime, assuming
that the time dimension of the query is small compared to the time window of the database,
would be done by calculating a minimal covering of the spatial region with the square boxes.
Then a range query could be placing for each box id between the desired start and end timestamps.
Several inefficencies may result from this method:
\begin{enumerate}
\item If the boxes are too small relative the query area, the number of queries (each
of which will lead to increased overhead) may be very large. While these can be placed
in parallel, there will eventually be an I/O bottleneck on the machines storing
the database when the number of parallel queries is large.
\item If the query area is smaller than the size of a box, the resultant queries
will have a large number of false positives: points returned which are not in the
desired query region. This leads to significant increases in both computational and I/O costs.
\item Even when the number of covering boxes is neither too small nor too large, the parallel
execution of the queries over each box will not necessarily be balanced. The speedup will be
sub-linear, with the query time driven by the densest box which commonly has an order of
magnitude more data than the average data volume per block.
\end{enumerate}
Issues (1) and (2) are dual to one another, and if all queries have similarly sized
spatial regions the database can be constructed to use boxes with an area optimized
for these queries. However, in the typical case that queries of varying sizes must be
supported, any choice will run into one of these problems. Regardless, the load
balancing associated with (3) continues to be a problem for any dataset with significant
spatial bias.

As a solution to these problems, we propose a method for partitioning space as regions
with equal data volumes rather than boxes with equally sized volumes. This is done by
modifying the standard geohash algorithm without sacrificing any main benefits of the
geohash algorithm or substantially increasing the computation burden of encoding or
decoding input data points.

\section{Related Work} \label{sec:related}

There is a long history of work on data structures for
optimizing geospatial and geospatial-temporal queries such
as $k$-nearest neighbors and spatial joins. Common examples
include R-trees \cite{guttman1984r},
$R^{*}$-trees \cite{beckmann1990r}, and quadtrees \cite{samet1985storing}.
These algorithms often offer superior performance to simple space filling
curves by replicating the true dimensionality of the problem by an abstraction
of trees structures as linked lists. Unfortunately linked lists do not fit into
the primary key design scheme of the distributed databases mentioned in
the previous section, making these more sophisticated algorithms out of
reach for productionalized systems.

Some research has been done towards producing a distributed database which would
be able to natively implement the geospatial temporal database structures which have been successful
on single machine scale applications. These include VegaGiStore \cite{zhong2012towards},
VegaCache \cite{zhong2013vegacache} and VegaCI \cite{zhong2012distributed}, \cite{zhong2012elastic}.
However at this time these remain merely theoretical as no distribution of any of these
has even been publicly distributed, let alone developed to the maturity needed for a
production system.

The simplier approach of adapting geohashes by interleavings spatial and temporal dimensions
has also been explored. Given that these can be quickly implemented on top of existing industry
standard databases this is the approach we have choosen to augment with our work. Fox et al.
developed a generic approach for blending these two dimensions, along with methods for
extending to lines and polygons and for producing fast queries in Apache Accumulo \cite{fox2013spatio}.
The HGrid data model is a similar application which takes advantage of the specific secondary
indicies and related optimizations present in HBase \cite{han2013hgrid}. Both rely on
specifying the relative resolution of the spatial and temporal dimensions, and are able to
use our entropy maximizing geohash with no major changes.

Finally, we note that the desire for a data balanced design which motivated our work is a
natural parallel to the problems associated with efficent query design in the
related tree-based work. Red–-black trees \cite{bayer1972symmetric} and the
Day–-Stout–-Warren algorithm for binary trees \cite{stout1986tree} are important examples
of load balancing used in relation to tree-based data structures.
In this context, the entropy maximizing geohash can be thought of as the curve filling
analogue to balanced trees.

\section{Entropy Balanced Geohash}  \label{sec:geohash}

\subsection{A Formulation of the Standard Geohash Encoding}

As mentioned, a geohash is a simple scheme for mapping two-dimensional coordinates into a
hierarchical, one-dimensional encoding. It is explained in several other sources, but we
re-construct it here in a format which will be most conducive to generalizations.
The first step is to map latitude and longitude coordinates in a
standard unit square; this is done by the following linear mapping:
\begin{align}
x &= \frac{\text{lon} + 180}{360} \\
y &= \frac{\text{lat} + 90}{180}
\end{align}
This choice is by convention, and any other method for mapping coordinantes
into the unit square is equally viable.

The $x$ and $y$ coordinates need to be expressed in as a binary decimals.
Formally, we define the $x_i \in \{0,1 \}$ and $y_i \in \{0,1\}$
(with the restriction that neither is allowed to have an infinite trailing
tail with all `1's, for uniqueness) such that
\begin{align}
x &= \sum_{i=1}^{\infty} \frac{x_i}{2^i},\\
y &= \sum_{i=1}^{\infty} \frac{y_i}{2^i}.
\end{align}
A geohash representation of $(x,y)$ is constructed by interleaving these
binary digits. The $q$-bit geohash $g_q(\cdot,\cdot)$ can symbolically
be defined as
\begin{align}
g_q(x,y) &:=  \sum_{i=1}^{\lceil q/2 \rceil} \frac{x_i}{2^{2i-1}} +
              \sum_{i=1}^{\lfloor q/2 \rfloor} \frac{y_i}{2^{2i}}.
\end{align}
It is fairly easy to show that the geohash function is monotone increasing
in $q$, with the growth strictly bounded by $2^{-q}$, so that
\begin{align}
0 \leq g_{q+m}(x,y) - g_{q}(x,y) < \frac{1}{2^q}
\end{align}
For all $m$ greater than zero.

\subsection{Entropy}

A geohash is typically used as an index in the storing and querying of large
spatial processes. A simple theoretical model for a stream of spatial data can
be constructed by assuming that each observation is an independent identically
distributed random variable $\mathfrak{F}$ from some distribution over space.
Borrowing a concept from information theory, we can define the entropy of a geohash
over a given spatial distribution by the equation
\begin{align}
H(g_q) &:= -1 \sum_{v \in \mathcal{R}(g_q)} \mathbb{P} \left[ g_q(\mathfrak{F}) = v\right] \cdot
            \log_2 \left\{ \mathbb{P} \left[ g_q(\mathfrak{F}) = v\right] \right\} \label{entropyDef}
\end{align}
Where $\mathcal{R}(g_q)$ is the range of the $q$-bit geohash. It is a standard result
that the entropy of a discrete distribution is maximized by the uniform distribution.
Therefore we can use the entropy per bit as a proxy for how balanced a geohash is
for a given distribution of spatial data.

\subsection{The Generalized Geohash}

As the $q$-bit geohash function is bounded and monotonic in $q$, we can define the infinite
precision geohash, which we denote as simply $g(\cdot, \cdot)$, to be the limit
\begin{align}
\lim_{q \rightarrow \infty} g_q(x,y) &:= g(x,y).
\end{align}
With this continuous format, one can see that if we compose $g$ with an appropriate new function
$h$, the composition $h \circ g(x,y)$ can be thought of as a rescaled version of the
traditional geohash. To be precise, we would like a function $h$ to have the following properties:
\begin{align}
&h: [0,1] \rightarrow [0,1], \label{hDef} \\
&h(0) = 0, \\
&h(1) = 1, \\
&x < y \iff h(x) < h(y). \label{monotoneEq}
\end{align}
Note that Equation~\ref{monotoneEq} implies that $h$ is also continuous. From here, we can define
the analogue to a $q$-bit geohash by truncating the binary representation of $h(z) = w$,
\begin{align}
h(z) &= \sum_{i=1}^{\infty} \frac{w_i}{2^i}
\end{align}
To the its first $q$-bits
\begin{align}
h_q(z) &:= \sum_{i=1}^{q} \frac{w_i}{2^i}.
\end{align}
In the remainder of this paper, we refer to $h_q \circ g(x,y)$ as a {\it generalized geohash}.

\subsection{The Empirical Entropic Geohash}

We have introduced the concept of a generalized geohash in order to construct a spatial encoding
scheme which better optimizes the entropy as defined in Equation~\ref{entropyDef}.
Assume $\{z_i\}_{i=0}^N$ is a set of independent samples from realizations of the random variable
$\mathfrak{F}$. The empirical cumulative distribution function $G$ of the standard geohash function
$g(\cdot, \cdot)$ is given by
\begin{align}
G(t) &:= \frac{1}{N} \cdot \sum_{i=0}^{N} 1_{g(z_i) \leq t}, \, t \in [0,1]. \label{empGeohash}
\end{align}
From Equation~\ref{empGeohash} we can define the {\it entropy maximizing geohash} function
$b$ (balanced), assuming that every point $z_i$ has a geohash $g(z_i)$ strictly between $0$ and
$1$, to be
\begin{align}
b^q(t) :=
  \begin{cases}
    0 &\mbox{if } t = 0\\
    1 &\mbox{if } t = 1\\
    \frac{N}{N+2} \cdot G^{-1}(t) &\mbox{if } \, \exists \, i \in \mathbb{Z} \, \text{s.t.} \, t = i / 2^{q} \\
    \text{linear interpolation of the above points} & \text{else}
  \end{cases} \label{balancedHash}
\end{align}
The balanced geohash is essentialy the inverse of the empirical function $G$, with some minor
variations to satisfy Equations~\ref{hDef}-\ref{monotoneEq}.
If the points $\{z_i\}$ are unique, and $N$ is sufficently large, the $q$-bit analogue $b_q^q$
of Equation~\ref{balancedHash} will have an entropy $H(b_q^q)$ of approximately equal to $q$.

More formally, we can prove the following bound on the entropy of the balanced geohash:
\begin{thm} \label{entropyThm}
Let $b_q^q$ be the entropy balanced geohash estimated from a sample of $N$ unique data points. Then,
with probability at least $1 - 2 e^{-0.49 \cdot 2^{-2q} N \cdot [1 - A]^2}$ the entropy
$H(b_q^q)$ is bounded by the following simultaneously for all values $A \in [0,1]$:
\begin{align}
H(b_q^q) &\geq q \cdot \frac{N}{N+2} \cdot A
\end{align}
\end{thm}
\begin{proof}
Let $F(\cdot)$ be the true cumulative distribution function of the variable $\mathfrak{F}$,
and $F_N(\cdot)$ be the empirical cumulative distribution function from a sample of $N$
independent observations. Setting $\epsilon = [1 - A] \cdot \frac{N}{N+2} \cdot 2^{-(q+1)}$,
the Dvoretzky–-Kiefer–-Wolfowitz inequality states \cite{dvoretzky1956asymptotic} that
the following holds for all values
$A \in [0,1]$ with probability $1 - e^{-2N\epsilon}$:
\begin{align}
|F(x) - F_n(x)| \leq [1 - A] \cdot \frac{N}{N+2} \cdot 2^{-(q+1)}
\end{align}
Therefore, the empirical entropy should be bounded as follows:
\begin{align}
H(b_q^q) &\geq -1 \cdot \sum_{i=0}^{2^q-1} \left[F(i/2^q) - F((i+1)/2^q)\right] \times
  \log_2 \left[F(i/2^q) - F((i+1)/2^q)\right] \\
&\geq -1 \cdot \sum_{i=0}^{2^q-1} \left[F_n(i/2^q) - F_n((i+1)/2^q) - 2\epsilon \right] \times
  \log_2 \left[F_n(i/2^q) - F_n((i+1)/2^q) - 2\epsilon \right] \\
&= -2^{q} \cdot \left[\frac{N}{N+2} \cdot \frac{1}{2^q} - 2 \epsilon \right] \times
  \log_2 \left[\frac{N}{N+2} \cdot \frac{1}{2^q} - 2 \epsilon  \right] \\
&= -1 \cdot \frac{N}{N+2} A \times \log_2 \left[\frac{N}{N+2} \cdot \frac{1}{2^q} \cdot A  \right] \\
&\geq q \cdot \frac{N}{N+2} A
\end{align}
Which, plugging in the appropriate $\epsilon$ into the probablity bound, yields the desired result.
\end{proof}
Plugging in $A=2/3$, the bound in Theorem~\ref{entropyThm} holds with probability greater than $0.99$
for a 5-bit balanced geohash when $N$ is at least $1e5$ and for a 10-bit geohash when $N$ is at
least $1e8$. We will see in the our empirical examples that the rate of convergence of the entropy
to $q$ is typically much faster.

\section{Census Data Example}  \label{sec:census}

\begin{figure}
\centering
\includegraphics[width=\textwidth]{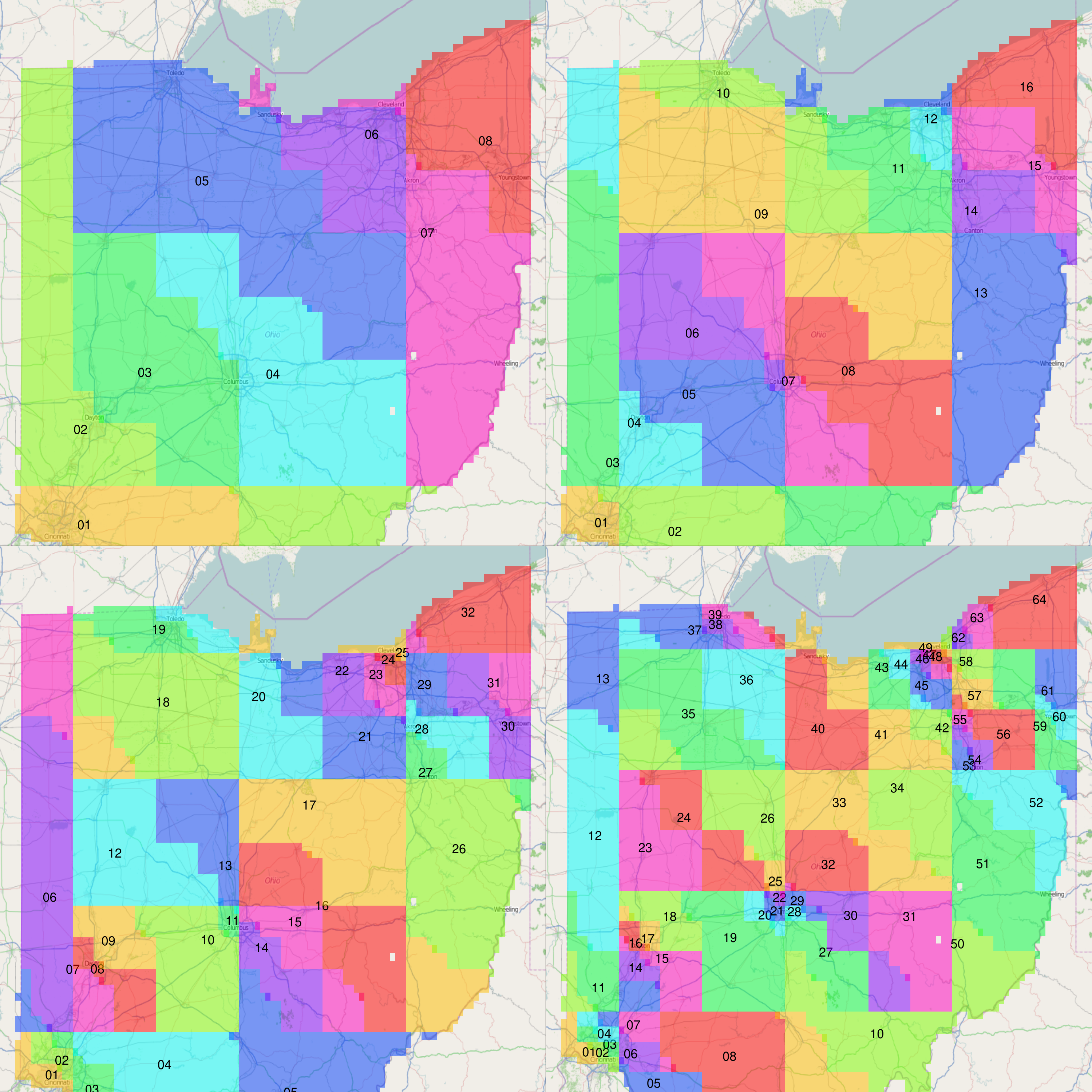}
\caption{Example entropy balanced geohash buckets (those which have a common prefix)
of $3$,$4$,$5$, and $6$-bits learned from population data in the state of Ohio.
Numbers are sequentially ordered from the smallest to the largest geohash buckets,
and the buckets were truncated to the state boundaries to increase readability
even though they theoretically extend outside as well.}
\label{fig:ohioHash}
\end{figure}

\begin{figure}
\centering
\includegraphics[width=6in]{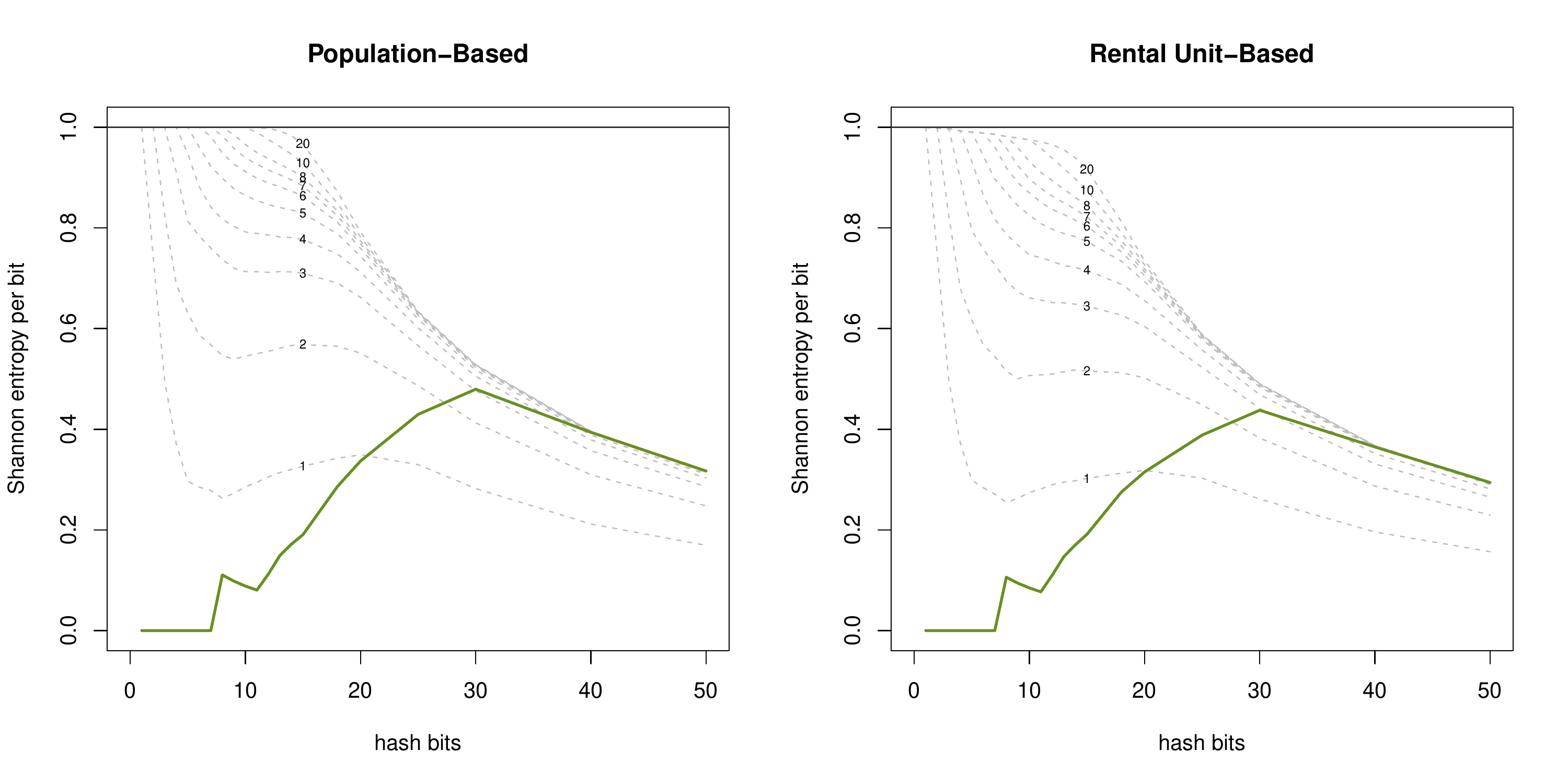}
\caption{Average entropy for population and rental household data in Ohio. The green
line is the standard geohash; dashed lines are balanced geohashes, and the overlaid
numbers give the number of bits to which is is balanced. The rental unit balanced
geohashes are balanced on population data rather than the rental data.}
 \label{fig:ohioEntropy}
\end{figure}

\begin{figure}
\centering
\includegraphics[width=6in]{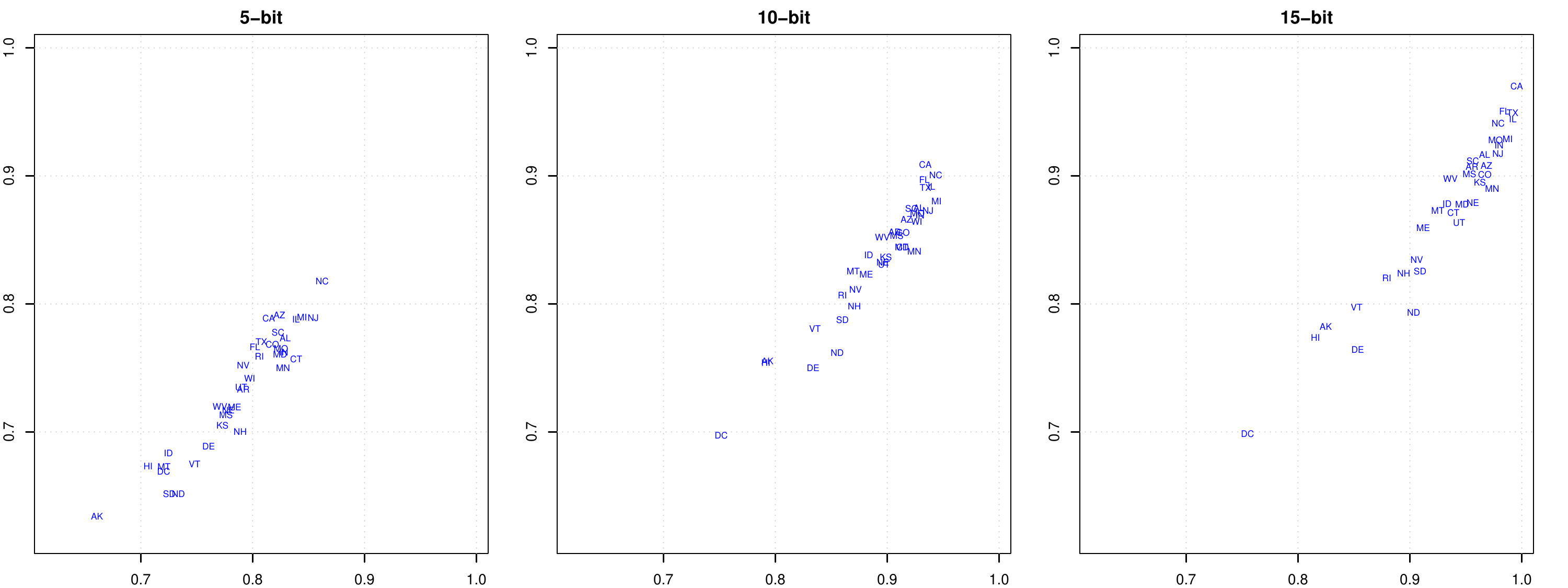}
\caption{Plots of entropy per bit for population (x-axis) versus rental units (y-axis)
for 5,10, and 15-bit balanced geohashes. Only 34 states (and the District of Columbia)
are shown because overplotting otherwise made the image too difficult to read.}
 \label{fig:allStatesEntropy}
\end{figure}

Many large geospatial datasets are distributed roughly proportional to population. Log files
geocoded by ip addresses for instance will tend to have more activity in dense urban areas
and less in rural regions. We use aggregate count data from the US Census data as a natural
example for exploring the balanced geohash because these counts are also proportional to
population (exactly equal, in fact). The Census Bureau distributes decennial aggregate
counts across thousands of attributes at various levels of grainularity. The smallest region
with un-sampled counts are
census blocks, small geographic regions containing the residences of typically somewhere
between 600 and 1300 people. There were over 11 million such regions in the 2010 Census. We
use census blocks here in order to have the maximum level of geospatial grainularity available.
Due to the smaller aggregation size, only a small subset of summary statistics are provided
for census blocks. The total population will be used to construct the balanced geohash, with
the total number of households renting their primary residence (contrasted with owning) being
used to show the relative robustness of the balanced hash.

To help visualize the balanced geohash, the entropy balanced hash up to $3$,$4$,$5$, and
$6$-bits were calculated from the population data in the state of Ohio. In Figure~\ref{fig:ohioHash}
the geohash `buckets', all points in space with a common prefix, are illustrated over a
map of the state. Technically the buckets extend over the entire globe, but are truncated here
to increase readability of the graphic. Notice that each region in the $q$-bit bucket is bifurcated
into two regions in the coorisponding $q+1$-bit bucket. Also, as is common, the smaller set of
$3$-bit buckets all have roughly the same area. However, the $6$-bit buckets differ vastly in size
such as the small area of region 21 and large area of region 51. As expected, the buckets
tend to clump up around the major cities in Ohio, such as the $6$-bit buckets 20-22, 28, 29
around Columbus.

Continuing with the census example in Ohio, Figure~\ref{fig:ohioEntropy} plots the average
entropy for a standard geohash and several balanced geohashes for census block population
and rental household data. The rental household data uses a geohash balanced on population,
and illustrates the degree to which our method is robust to changes in the underlying data
model. In both cases, as expected, the balanced geohashes have higher average entropies the
when more bits are used for balancing. There are significantly diminishing returns of
increasing $q$, however,
with the improvement from 2 to 3 bits being larger than the jump from 5 to 20. The entropy
from the rental unit-based data is generally quite similar to the population-based data,
indicating that the balanced hash is fairly robust to moderate changes in the generating
model. Rental units are correlated with population, but do exhibit spatial dependence with
a higher rate of renting in more urban and less affluent areas as well as spikes near
large universities.

With significantly higher levels of precision, across all of the hashes, the entropies
converge and the entropy per bit begins to decay to zero. The reason for this is that
there are less than $2^{19}$ census blocks in the entire state of Ohio, and even the standard
geohash manages to put every block's centroid into its own bucket after using at least
$35$-bits. Ultimately the discreteness of the data limits our ability to further increase
the entropy even with additional splits. Additionally, the $20$-bit balanced geohash would
in theory have a near perfect entropy per bit up to $20$-bits of hash length if the data were
continuous. In fact, the $20$-bit geohash at $15$-bits of precision (note that this is
indistingusihable from a $15$-bit geohash) has only and entropy of $0.9$ per hash pbit.

Figure~\ref{fig:allStatesEntropy} plots the entropy of three balanced geohashes for
$34$ other states at $15$-bits of precision. The leftmost panel shows the outcome of
only balancing on the first $5$-bits. The worst states here are Alaska, Hawaii, South
Dakota, and North Dakota. These are all states which have a very low population density
(for Hawaii, the density is low when averaged over the area of a bounding box for the
entire island chain). The rightmost panel shows the pure effects
of the discreteness of the data, as all of the x-values would otherwise be nearly perfect
values near $1$. Not surprisingly, the most populous state of California fairs the best
followed closesly by other large states such as Texas and Florida. The smallest states
have the worst entropy here, with DC being the worst at only $0.75$ average entropy per bit.
In all cases, the hash balanced on population performs similarly in regards to the number
of rental properties.

\section{Future Extensions}  \label{sec:future}

We have presented an extension on the traditional geohash which is balanced on the data
volume rather than boxes of constant latitude and longitude. The benefits of this approach
have been shown for both spatial and spatial-temporal databases. Theorem~\ref{entropyThm}
establishes the robustness of our model-based approach due to random noise and an empirical
study of US Census data demonstrates robustness to model drift. The balanced geohash is
usable for indexing large databases as-is; in most applications it will perform most queries
in less time, with fewer scans, and lower overhead compared to a standard geohash.
There are several extensions of our work which would further increase the applicability
of the entropy balanced geohash.

Our balanced geohash $b^q_m(\cdot)$ has two free parameters which give the number of bits for which
the hash is balanced ($q$) and length of the entire hash ($m$). The latter is largely
application specific, but the choice of $q$ has some important efficency aspects. If $q$ is
too large the geohash can become quite noisy and overfit to the sample of data used to
learn the underlying model. Additionally, in order to map a standard geohash to a balanced
geohash we need to know the $2^q$ break points to the linear interpolation given in
Equation~\ref{balancedHash}. So if $q$ is large the memory overhead of the balanced
geohash may become prohibitive. In order to better estimate an optimal $q$, tighter two-way
theoretical bounds are needed as compared to those given in Theorem~\ref{entropyThm}. This
will likely require replacing the distribution-free approach employeed in our proof with a
specific family of models such as Gaussian mixtures. In the meantime we can safely pick a
small $q$ without concern, knowing that the noise level is low but that long hash prefixes
may be mildly unbalanced.

The paper \cite{han2013hgrid} provides a generic approach for interleaving geospatial
and temporal components utilizing the specific implementation optimizations present in
HBase. As mentioned, the geohash used in their paper can be seamlessly replaced by the
entropy balanced geohash. The exact method of choosing how to incorporate the time and
space dimensions are not specified as this will largely be a function of the specific
underlying dataset. In the case of the balanced geohash, since the data boxes are
already balanced by the distribution of the data, it should be possible to at least
specify a methodology for choosing how to interleave these components (with perhaps
one or two tuning parameters to tweak based on the specific use-cases). For now the
method of combining these two dimensions can be done fairly well manually, particularly
if the scheme is kept simple.

Finally, while the balanced geohash is fairly robust to model drift, a method akin to
self-balancing trees whereby the balanced geohash could adapt over time would be a
great addition to fixed hash presented here. This would be particularly important when
implementing a large $q$-bit balanced geohash. Possible methods for doing this include
adaptively splitting large prefixes and joining small prefixes; this would require
re-writting data, but would hopefully not be required very often. Alternatively a
scheme could be used to only balance new data as it arrives. In either case, care would
need to be taken to avoid large increases to the computational complexity or memory
overhead of an adaptive modification to the entropy maximized geohash design.

\nocite{*}
\bibliography{bhash}

\end{document}